\documentclass{article}
\usepackage[dvips]{graphicx}
\usepackage{amsmath,amssymb,amsthm,bm,fancybox,fullpage}
\usepackage{bm}
\usepackage{cite}
\usepackage{verbatim} 
\usepackage{url}
\usepackage{setspace}
\newtheorem{defi}{Definition}{\bfseries}{\itshape}
\newtheorem{theo}{Theorem}{\bfseries}{\itshape}
{\bfseries}{\itshape}
\newtheorem{prop}{Proposition}{\bfseries}{\itshape}
\newtheorem{exam}{Example}{\bfseries}{\itshape}
{\bfseries}{\itshape}
{\bfseries}{\itshape}
\newtheorem{lem}{Lemma}{\bfseries}{\itshape}

\newlength{\ralen}
\newcommand{\back}{\raisebox{\ralen}{\framebox(11,13){\bf\large ?}}}
\newcommand{\card}[1]{\raisebox{\ralen}{\framebox(11,13){#1}}}

\author{
Yuji Hashimoto\thanks{Tokyo Denki University, Japan ({\tt 17rmd23@ms.dendai.ac.jp})}, 
Kazumasa Shinagawa\thanks{Tokyo Institute of Technology, Japan ({\tt shinagawa.k.a@m.titech.ac.jp})}, 
Koji Nuida\thanks{National Institute of Advanced Industrial Science and Technology (AIST), Japan ({\tt k.nuida@aist.go.jp})}, 
Masaki Inamura\thanks{Tokyo Denki University, Japan ({\tt minamura@rd.dendai.ac.jp})}, 
Goichiro Hanaoka\thanks{National Institute of Advanced Industrial Science and Technology (AIST), Japan ({\tt hanaoka-goichiro@aist.go.jp})}
}
\date{\today}

\begin{document}
\title{Secure Grouping Protocol Using a Deck of Cards}

\maketitle
\begin{abstract}
We consider a problem, which we call \emph{secure grouping}, of dividing a number of parties into some subsets (groups) in the following manner: Each party has to know the other members of his/her group, while he/she may not know anything about how the remaining parties are divided (except for certain public predetermined constraints, such as the number of parties in each group).
In this paper, we construct an information-theoretically secure protocol using a deck of physical cards to solve the problem, which is jointly executable by the parties themselves without a trusted third party.
Despite the non-triviality and the potential usefulness of the secure grouping, our proposed protocol is fairly simple to describe and execute.
Our protocol is based on algebraic properties of conjugate permutations.
A key ingredient of our protocol is our new techniques to apply multiplication and inverse operations to hidden permutations (i.e., those encoded by using face-down cards), which would be of independent interest and would have various potential applications.
\end{abstract}

\section{Introduction}

\emph{Multiparty computation} (MPC) is a cryptographic technology that enables two or more parties to jointly compute a multivariate function from their local inputs, in such a way that each party knows the party's local input/output pair but may not know anything about other parties' local inputs and outputs except for those implied by the party's own input/output pair only.
A direction in the study of MPC, which has recently been an active branch in this area, is so-called \emph{card-based protocols}~\cite{Boer89,Kilian93,KochAC15,MizukiCANS16,MizukiUCNC13,MizukiAC12,MizukiIJISEC14,MizukiFUN14,MizukiFAW09,IshikawaUCNC15,Uchiike06,NakaiCANS16,Niemi98,NishidaTAMC15,NishidaTPNC13,NishimuraAsiaPKC16,ShinagawaProvSec15,ShinagawaIWSEC15,ShinagawaIEICE16,Stiglic01,UedaTPNC16}, where protocols for MPC are supposed to use a deck of \emph{physical cards} instead of usual electronic computers.
In a card-based protocol, private information is usually encoded by using face-down cards with mutually indistinguishable back sides, and randomness is introduced by applying \emph{shuffle operations} to some face-down cards.
A typical property is that, in contrast to ordinary computer-based MPC where each party may execute a program at local environment (hence the security has to rely on certain cryptographic techniques, some of which may be only computationally secure), a card-based protocol is supposed to be executed at a public place where the parties can simply monitor and prevent the other parties' adversarial behaviors without any cryptographic machinery.
Consequently, it is usual that card-based protocols provide information-theoretic security.

For card-based protocols, it is known that every function is at least securely computable when ignoring possibly expensive computational costs\cite{Kilian93,MizukiFAW09}.
On the other hand, many efficient card-based protocols specialized to some typical problems have been also investigated.
In those previous studies, the target problem to be solved by card-based protocols was usually a type of problem that already had an efficient computer-based counterpart, such as the case of card-based Millionaires' Problem~\cite{NakaiCANS16}; see the Related Works paragraph below for further details.
In contrast, in this paper we deal with a new type of interesting problem described below, which we call \emph{secure grouping}; for this problem, even a computer-based solution (except ones yielded by naively applying general-purpose MPC protocols) has not been known to the authors' best knowledge.

The secure grouping is defined as the problem of dividing a number of parties into some subsets (called \emph{groups}) in the following manner: Each party has to know the other members of his/her group, while he/she may not know anything about how the remaining parties are divided, except for certain public predetermined constraints such as the number of parties in each group.
For instance, suppose that there are six parties, say, Parties $1,2,\dots,6$, and they wish to randomly divide themselves into three pairs.
Some examples of the possibilities are ($12$, $34$, and $56$), ($14$, $26$, and $35$), ($16$, $23$, and $45$), etc.
Then the goal is to generate one of the all possibilities uniformly at random, while each party has to know who is the partner but may not know about the other two pairs.

It is worth emphasizing that such a secure grouping cannot be achieved by a simple lottery; namely, when each of the six parties in the example above picks up one of the two $\card{$\heartsuit$}$'s, two $\card{$\clubsuit$}$'s, and two $\card{$\diamondsuit$}$'s, there seems to be no simple way for every party to know the other party having the same card without revealing any party's card to the remaining parties.
This suggests that secure grouping is really a non-trivial problem.
We also note that our setting of secure grouping covers various situations, such as the case where $n$ parties wish to randomly select two distinguished persons (like \lq\lq Werewolves'' in the famous Werewolf game) in such a way that only the distinguished persons themselves know who are the distinguished persons; or the dealer in a card game wishes to randomly select a partner from the other players in such a way that only the dealer and the partner him/herself know who is the dealer's partner\footnote{In some card games, the dealer announces one of the cards (e.g., \lq\lq $\spadesuit\,8$'') and then the player having this card becomes the dealer's partner.
However, now the dealer cannot know who is the partner, though the partner him/herself can know that he/she is the dealer's partner; hence the condition of secure grouping is not achieved.}.
The flexibility of secure grouping would be interesting and be potentially useful.

\paragraph{\bf Our Contributions.}

In this paper, we propose a card-based protocol to solve the problem of secure grouping explained above.
As opposed to usual card-based protocols where two kinds of cards (e.g., $\card{$\heartsuit$}$ and $\card{$\clubsuit$}$\,) are used, here we use different cards (with indistinguishable back sides) whose front sides are numbers $\card{$1$}\,, \card{$2$}\,, \dots$, which we call \emph{number cards}.
A face-down card with front side $\card{$k$}$ is called a \emph{commitment} of $k$.
By a rough estimate, our proposed protocol uses approximately $3dn$ number cards where $n$ is the number of parties to be divided into groups and $d$ is the maximal number of parties in a group.

One of our main ideas is to utilize some algebraic properties of conjugate permutations (see Section \ref{permutation_randomizing_protocol} for details).
To intuitively explain, here we consider a case of dividing seven parties into two pairs and one triple.
In this case, we deal with permutations of $1,2,\dots,7$, where a permutation $\sigma$ is encoded as the sequence of number cards with front sides $\sigma^{-1}(1),\sigma^{-1}(2),\dots,\sigma^{-1}(7)$\footnote{Note that this sequence of number cards is obtained by moving, for each $k = 1,2,\dots,7$, the $k$-th card $\card{$k$}$ to the $\sigma(k)$-th position.  For example, if $\sigma(k) = k+1$ for $1 \leq k \leq 6$ and $\sigma(7) = 1$, then the resulting card sequence is $\card{$7$}\,\card{$1$}\,\card{$2$}\,\card{$3$}\,\card{$4$}\,\card{$5$}\,\card{$6$}$.}.
Now we note that a grouping like ($ab$, $cd$, and $efg$) can be represented by a permutation of the form $\tau = (a,b)(c,d)(e,f,g)$, which means that $\tau$ exchanges $a$ and $b$, exchanges $c$ and $d$, and changes $e,f,g$ cyclically to $f,g,e$, respectively.
Then the problem of secure grouping is reduced to generating uniformly at random, in a committed form (i.g., each number card is faced down), a permutation $\rho$ of the same \lq\lq type'' $(*,*)(*,*)(*,*,*)$ and also the square $\rho^2$ of the permutation; once commitments of such $\rho$ and $\rho^2$ are obtained, each party, say Party $i$, can know the other two (or fewer) parties in his/her group by picking up the $i$-th face-down cards for $\rho$ and $\rho^2$.
For example, when $\rho = (1,5)(3,6)(2,7,4)$, the commitments to $\rho$ and $\rho^2$ are given by
\[
\rho = \card{$5$}\,\card{$4$}\,\card{$6$}\,\card{$7$}\,\card{$1$}\,\card{$3$}\,\card{$2$} \quad\mbox{and}\quad \rho^2 = \card{$1$}\,\card{$7$}\,\card{$3$}\,\card{$2$}\,\card{$5$}\,\card{$6$}\,\card{$4$}
\]
(where each card is actually faced down), and then
\begin{itemize}
\item
for example, Party $4$ obtains $\card{$7$}$ and $\card{$2$}$, which tells that Parties $7$ and $2$ are the other members of the group of size $3 = 2 + 1$;
\item
while Party $6$ obtains $\card{$3$}$ and $\card{$6$}$ (the party's own number), which tells that Party $3$ is the other member of the group of size $2 = 1 + 1$.
\end{itemize}
We note that, when the sizes of the groups are at most $d$, a similar process can be done by using permutations $\rho,\rho^2,\dots,\rho^{d-1}$.
Moreover, group theory ensures that the process of randomly shuffling the seven numbers appearing in a given permutation $\tau$ without changing the type is equivalent to taking a conjugate permutation $\sigma^{-1} \tau \sigma$ with random permutation $\sigma$ of the seven numbers.
Then the latter problem can be solved by using a protocol for computing multiplication and inverse of permutations in a committed form; this protocol (see Section \ref{permutation_randomizing_protocol} for details) is also a part of our contribution in this paper, which would be of independent interest.
Secure grouping is now achieved by combining these ideas.
See Section \ref{secure_grouping_protocol} for details.

The \lq\lq plain'' protocol explained above is seemingly applicable only to \lq\lq simple'' types of secure grouping where the parties have \lq\lq symmetric'' roles and the groups with the same number of members have \lq\lq symmetric'' roles as well.
Nevertheless, in fact the idea of the protocol is also applicable to more complex types of secure grouping.
For example, in the aforementioned case of selecting two distinguished persons, we can use secure grouping of type $(*,*)(*)(*) \cdots (*)$ and then the only group with two members specifies the two distinguished persons.
On the other hand, in another aforementioned case of choosing a partner of the dealer (numbered as Player $1$), we can use our secure grouping protocol starting from a permutation $(1,2)(3)(4) \cdots (n)$ and then shuffling all the numbers \emph{except the number $1$} (i.e., the random permutation $\sigma$ is chosen with constraint $\sigma(1) = 1$); now the resulting permutation $\rho$ is of the form $(1,k)(*)(*) \cdots (*)$, the number $k$ on the card picked up by Player $1$ (dealer) specifies the partner, and the partner will pick up the card $\card{$1$}$ which tells that he/she is the dealer's partner.
Moreover, we can also handle the cases where the groups with equal numbers of parties have to be mutually distinguished, by appropriately introducing some dummy number cards indicating the \lq\lq names'' of groups and then shuffling all the numbers except for dummy numbers.
These examples show the flexibility of our proposed protocol.

\paragraph{\bf Related Works.}
It is known that every function can be securely computed based on a deck
of cards~\cite{Kilian93,MizukiFAW09}.
Besides researches for improving general-purpose protocols, the other
important direction is to investigate efficient card-based protocols customized to some useful applications: for
example, the problem of generating secret permutations without fixed points ~\cite{Kilian93,IshikawaUCNC15},
\emph{secure voting}~\cite{MizukiUCNC13,ShinagawaProvSec15}, and
\emph{Millionaires' Problem}~\cite{NakaiCANS16}.
In the early research of card-based cryptography, Cr{\'{e}}peau and
Kilian~\cite{Kilian93} constructed a protocol that randomly selects a permutation with no fixed point
without revealing which one was selected.
It has an application for e.g., exchanging gifts among multiple players in
which each player does not receive his/her own gift.
Recently, Ishikawa et al.~\cite{IshikawaUCNC15} introduced a new shuffle
called a \emph{Pile-Scramble Shuffle} to improve the
protocol in \cite{Kilian93 }.
We use the Pile-Scramble Shuffles in the construction of our protocols.
For the secure voting, Mizuki, Asiedu, and Sone~\cite{MizukiUCNC13}
constructed a protocol for two candidates, which takes $n$ bits as
inputs and outputs the sum of the inputs.
Recently, Shinagawa et al.~\cite{ShinagawaProvSec15} constructed a
secure voting protocol for multiple candidates based on a new type of
cards .
For the Millionaires' Problem, Nakai et al.~\cite{NakaiCANS16} constructed
a protocol, which takes two strings $x, y$ as inputs and outputs a bit indicating whether $x > y$ or not.

\section{Preliminaries}

In this section we prepare necessary tools to construct our secure grouping protocol.                   
We suppose that a distinct number from $1$ to $n$ is assigned to each player in advance, where $n$ is the total number of players, and the correspondence between the numbers and the players is publicly known.
We identify a player with the assigned number. 
Throughout this paper, $S_n$ denotes the group of permutations on the set $\{1,2,\dots,n\}$ of numbers. 

\subsection{Definitions and Properties about Permutations}\label{DTP}

In this subsection, we describe some definitions related to permutations and look at their properties. 

\begin{defi}[cyclic permutation]

A permutation $\tau$ is called a \emph{cyclic permutation} if there are a unique integer $r>1$ and distinct numbers $i_1,i_2,\dots,i_r$ satisfying the following conditions:
\begin{itemize}
\item
We have $\tau(i_1) = i_2,\dots,\tau(i_{r-1}) = i_r$, and $\tau(i_r) = i_1$.
\item
We have $\tau(k) = k$ for any number $k$ different from $i_1,i_2,\dots,i_r$.
\end{itemize}
In this case, we call the permutation $\tau$ a \emph{cycle of length $r$} and write it as $(i_1,i_2,\dots,i_r)$ (or simply $(i_1i_2 \cdots i_r)$ if no ambiguity occurs).
\end{defi}
In the case above, the set $\{i_1, i_2, \ldots, i_r\}$ is called the \emph{cyclic area} of the cyclic permutation $\tau = (i_1,i_2,\dots,i_r)$.
For example, the permutation $\tau \in S_4$ given by $(\tau(1),\tau(2),\tau(3),\tau(4)) = (1,4,2,3)$ is a cycle $(243)$ of length three with cyclic area $\{2,3,4\}$, while $\sigma \in S_4$ given by $(\sigma(1),\sigma(2),\sigma(3),\sigma(4)) = (2,1,4,3)$ is not a cyclic permutation.

We say that two cyclic permutations $\sigma,\tau$ with cyclic areas $C_{\sigma},C_{\tau}$, respectively, are \emph{disjoint} if $C_{\sigma} \cap C_{\tau} = \emptyset$.
For example, the two cyclic permutations $(123)$ and $(45)$ are disjoint, while $(264)$ and $(345)$ are not disjoint.
We note that disjoint cyclic permutations are commutative in the group of permutations.

The following fact about permutations is well-known.

\begin{prop}
Any permutation is uniquely represented by the product of disjoint cyclic permutations.
\end{prop}

For example, the permutation $\tau \in S_6$ given by $\tau(1) = 2$, $\tau(2) = 3$, $\tau(3) = 1$, $\tau(4) = 4$, $\tau(5) = 6$, and $\tau(6) = 5$ is decomposed into disjoint cycles as $\tau = (123)(56)$.
We also note that, it is convenient to consider as if a permutation $\sigma$ virtually involves \lq\lq cycle $(k)$ of length one'' when $\sigma(k) = k$; by using the abused notation, the permutation $\tau \in S_6$ above can be also represented by $\tau = (123)(4)(56)$.

Next we define the \emph{type of permutation}.
Type of permutation $\tau$ is the data of how many cycles of each length are present in the decomposition of $\tau$ into disjoint cycles as above. 
\begin{defi}[type of permutation]
Let $\tau \in S_n$, which is decomposed into disjoint cycles (including the virtual \lq\lq cycles of length one'' as mentioned above).
For each $i = 1,2,\dots,n$, let $m_i$ denote the number of cycles of length $i$ in the decomposition of $\tau$.
Then we say that $\tau$ is of \emph{type} $\langle 1^{m_1},2^{m_2},\dots,n^{m_n} \rangle$; here the terms $i^{m_i}$ with $m_i = 0$ may be omitted in the notation.
\end{defi}
Note that $\langle 1^{m_1}, 2^{m_2}, \ldots, n^{m_n} \rangle$ can be also viewed as the set of permutations of type $\langle 1^{m_1}, 2^{m_2}, \ldots, n^{m_n} \rangle$.
For example the permutation $\tau = (13)(25)(798) = (13)(25)(4)(6)(798) \in S_9$ belongs to the set $\langle 1^2, 2^2, 3^1 \rangle$. 

\subsection{Number Cards}
We use cards with numbers written on the front since these are convenient for treating permutations of numbers $1,2,\dots,n$ directly\footnote{
Usually, we define coding rules such as $\card{$\clubsuit$}\card{$\heartsuit$} = 0$ and $\card{$\heartsuit$}\card{$\clubsuit$} = 1$ 
since the card-based protocol normally uses Boolean values. 
If the usual Boolean encoding rule is used instead of the number cards, the secure grouping protocol can still be executed.
In the case the number of cards increases $2\lceil \log_2{n} \rceil $ times larger. }. 
We call the cards \emph{number cards} and write them as below.
\begin{align*}
\begin{array}{ccccc}
\card{1} & \card{2} & \cdots & \card{$n$}   
\end{array}
\end{align*}
The backs of number cards are indistinguishable.
We denote the back of a number card by \card{\back}\,.
A face-down card is called \emph{commitment}, and an operation to flip a face-down card into a face-up card is called \emph{open}.
Using the number cards, permutations in $S_n$ are represented by a card sequence $(x_1, x_2, \ldots, x_n)$ in a certain way explained later.

We also use the term ``permutation" as an \emph{operation} for card sequences. 
That is, we say ``applying a permutation $\sigma$ to a card sequence $x$" in the sense that rearranging $x$ according to $\sigma$, formally defined as follows. 

\begin{defi}[applying a permutation to a card sequence]
Let $\sigma \in S_n$ be a permutation and let $x = (x_1, x_2, \ldots, x_n)$ be a card sequence. 
We define a card sequence $\sigma(x)$ obtained by applying the permutation $\sigma$ to the sequence $x$ by 
\[
\sigma (x) := (x_{\sigma^{-1}(1)}, x_{\sigma^{-1}(2)}, \ldots, x_{\sigma^{-1}(n)}).
\]
\end{defi}
In other words this operation moves each $i$-th card to the $\sigma(i)$-th position.
For example, when $\sigma = (13)(265)(4)(7) \in S_7$ and $x = (x_1,\dots,x_7)$, we have $\sigma(x) = (x_3,x_5,x_1,x_4,x_6,x_2,x_7)$.
For the special case, the identity permutation ${\sf id}_n \in S_n$ is the identity operation such that a card sequence $(x_1, x_2, \ldots, x_n)$ is moved to a card sequence $(x_1, x_2, \ldots, x_n)$ itself.

\begin{defi}[card sequence representing a permutation]
Let $\sigma \in S_n$ be a permutation.  We define the card sequence for permutation $\sigma$ to be the card sequence $\sigma(\,\card{$1$},\card{$2$},\dots,\card{$n$}\,)$ obtained by applying $\sigma$ to the card sequence $x = (x_1,x_2,\dots,x_n)$ with $x_i = \card{$i$}\,$, $i = 1,2,\dots,n$.
\end{defi}
For example a permutation $\tau=(12)(34)(567)\in S_7$ is represented by the following card sequence
\begin{align*}
\begin{array}{ccccccc}
\card{2} & \card{1} & \card{4} & \card{3} &  \card{7} & \card{5} &  \card{6} \enspace .
\end{array}
\end{align*}

\subsection{Pile-Scramble Shuffle}

A \emph{shuffle}, which is an operation to apply a random permutation chosen from some distribution, plays an important role in card-based cryptography. 
While different types of shuffles are proposed and used for various applications, we use one of the shuffles called \emph{Pile-Scramble Shuffles}. 
It is proposed by Ishikawa et al. \cite{IshikawaUCNC15} and believed to be an ``efficient shuffle" since it has an easy implementation by e.g., utilizing physical envelopes.
\begin{defi}[Pile-Scramble Shuffle]
Let $n\geq 1$ be any integer.
The \emph{Pile-Scramble Shuffle of degree $n$} is the operation that takes a card sequence $x = (x_1,x_2,\dots,x_n)$ and outputs $r(x) = (x_{r^{-1}(1)}, x_{r^{-1}(2)}, \ldots, x_{r^{-1}(n)})$ where $r\in S_n$ is a random permutation and hidden from all parties.
\end{defi} 
Pile-Scramble Shuffle is described by using the following notation:
%, described as follows: 
\begin{align*}
\begin{array}{||c|c|c|c||c}
\card{\back} & \card{\back} & \ldots & \card{\back} & (x) 
\end{array}
\rightarrow
\begin{array}{ccccc}
\card{\back} & \card{\back} & \ldots & \card{\back} & (r(x)).
\end{array}
\end{align*}
We also define a similar operation for the case where each component $x_i$ of $x$ is not a single card but some other object, such as a collection of multiple cards.

\section{Permutation Randomizing Protocol}\label{permutation_randomizing_protocol}
In this section, we present a new protocol called \emph{permutation randomizing protocol} which is used as the main building block in our secure grouping protocol.
This section is our main technical contribution part.
In the simplest situation for our protocol, given an input permutation $\tau$ that is publicly known, this protocol outputs a committed card sequence representing a random permutation of the same type as $\tau$. %このアルゴリズムは、ある決められた置換の型を変えることなしにその型の集合の中から１つ置換を表すカード列を取り出すアルゴリズムである。
We emphasize that this functionality cannot be achieved by using naive shuffles since the Pile-Scramble Shuffle in general changes the type of a permutation. 
%しかし、対称群の元である置換を普通にシャッフルや置換によりランダマイズしてしまうと置換の型が変わってしまい、思い通りの置換を表すカード列を取り出すことができない。
Therefore, we need to realize an operation on permutations that does not change the type.
The key mathematical fact here is that any permutation $\rho$ that is conjugate to a permutation $\tau$ has the same type as $\tau$.
More precisely, we utilize the following well-known property in group theory:
%そこで、型を変えずに置換をランダマイズするアルゴリズムを提案する。数学的には、ある置換からそれと共役な置換へ移す写像を考えることにより、
%型を変えずにその型の中でのランダマイズを達成している。
  
\begin{lem}
\label{lem:expression_of_conjugation}
Let $\pi \in S_n$ be any permutation, which is expressed as the decomposition into disjoint cyclic permutations.
Let $\nu \in S_n$, and let $\pi'$ denote the permutation obtained by changing each number $j$ appearing in the expression of $\pi$ to the number $\nu^{-1}(j)$.
Then we have $\pi' = \nu^{-1} \pi \nu$.
\end{lem}
\begin{proof}
Let $a \in \{1,2,\dots,n\}$ and let $b := \pi'(a)$.
Then $b$ is (cyclically) next to $a$ in the expression of $\pi'$ as the decomposition into disjoint cyclic permutations.
By the definition of $\pi'$, this implies that $\nu(b)$ is (cyclically) next to $\nu(a)$ in the expression of $\pi$, which means that $\pi(\nu(a)) = \nu(b)$.
Hence we have $\nu^{-1}\pi\nu(a) = \nu^{-1}(\nu(b)) = b$, therefore $\pi'$ and $\nu^{-1}\pi\nu$ are equal as permutations.
\end{proof}

\subsection{Permutation Division Protocol}
Here we propose a protocol, called the \emph{permutation division protocol}, which is the main ingredient of our permutation randomizing protocol.    
%本節では、置換のランダマイズアルゴリズムの構成に必要不可欠なカードを用いた置換同士の演算アルゴリズムを構成する。
Given committed card sequences for permutations $v,w \in S_n$ as inputs, this protocol outputs the committed card sequence for permutation $v^{-1}w \in S_n$.
%このアルゴリズムは、置換コミットメント$v$と$w$が与えられたときに置換コミットメント$v^{-1}w$を生成するアルゴリズムである。
As explained later, this protocol enables us to generate a committed card sequence for a permutation $\sigma^{-1} \tau \sigma$ as in Lemma \ref{lem:expression_of_conjugation} from given card sequences for $\sigma,\tau$.

This protocol is composed of  four steps as follows.
Here, for any permutation $x$, we write $``(x)$" to mean that the displayed card sequence in a figure is the committed card sequence for $x$, while we also write $x$ to indicate that the displayed card sequence is the opened card sequence for $``x$".%このアルゴリズムは以下の4ステップからなる。
\begin{enumerate}
\item Arrange the committed card sequences for $v$ and $w$ as in the figure below. 
%置換コミットメント$v,w$を下図のように揃えて並べる。以降の図の表示において、カード列の右側に、置換の値を明示すると約束する。この際、裏返しの状態の置換コミットメントには丸括弧
%を用い、開示された状態の置換コミットメントには括弧を用いないものとする。
\begin{align*}
\begin{array}{ccccc}
\card{\back} & \card{\back} & \ldots & \card{\back} & (v) \\
\card{\back} & \card{\back} & \ldots & \card{\back} & (w) 
\end{array}
\end{align*}

\item Apply Pile-Scramble Shuffle to the first and the second rows simultaneously,
\begin{align*}
\begin{array}{||c|c|c|c||c}
\card{\back} & \card{\back} & \ldots & \card{\back} & (v) \\
\card{\back} & \card{\back} & \ldots & \card{\back} & (w) 
\end{array}
\rightarrow
\begin{array}{ccccc}
\card{\back} & \card{\back} & \ldots & \card{\back} & (rv) \\
\card{\back} & \card{\back} & \ldots & \card{\back} & (rw) 
\end{array}
\end{align*}
where $r \in S_n$ is a uniformly random permutation.     
%ここで、$r \in S_n$はランダムな置換である。

\item Open the first row, which reveals the permutation $rv$.
%上の行を開示し、$rv$を公開する。
Then apply the permutation $(rv)^{-1} = v^{-1}r^{-1}$ to the second row.
More precisely, the latter operation can be efficiently performed by rearranging the $n$ columns of the two rows in a way that the first row becomes the sequence $(1,2,\dots,n)$ representing $\mathsf{id}_n \in S_n$ where $\card{*}$ denote a face-up card having some $i \in \{1, 2, \ldots, n\}$.
\begin{align*}
&\begin{array}{ccccc}
\card{*} & \card{*} & \ldots & \card{*} & rv\\
\card{\back} & \card{\back} & \ldots & \card{\back} & (rw)
\end{array}
\\
\rightarrow
&\begin{array}{ccccc}
\card{1} & \card{2} & \ldots & \card{$n$} & \mathsf{id}_n\\
\card{\back} & \card{\back} & \ldots & \card{\back} & (v^{-1} r^{-1} rw)
\end{array}
\end{align*}

\item Output the second row (note that now $v^{-1} r^{-1} r w = v^{-1} w$). 
%下の行を出力する。
\begin{align*}
\begin{array}{ccccc}
\card{\back} & \card{\back} & \ldots & \card{\back} & (v^{-1}w)
\end{array}
\end{align*}

\end{enumerate}

The correctness of our protocol has been explained above.
On the other hand, the following property holds for the security of our protocol.

\begin{prop}
\label{prop:security_permutation_division}
The distribution of the only data available during the protocol, which is the card sequence for $rv \in S_n$ opened at Step 3, is uniform and is independent of $v$ and $w$.
\end{prop}
\begin{proof}
Indeed, for any $u \in S_n$, the number of the possible choice of the uniformly random $r$ that satisfies $rv = u$ is $1$ (i.e., $r = uv^{-1}$).
Hence, the permutation $rv$ appearing at Step 3 is uniformly random and independent of $v,w$, as desired.
\end{proof}

\subsection{Permutation Randomizing Protocol}
Here we describe our permutation randomizing protocol.
Given an input permutation $\tau$ that is publicly known, this protocol outputs a committed card sequence representing a random permutation of the same type as $\tau$.
In addition to the degree $n$ of permutations, our protocol in a general form takes an integer $k \geq 1$ (which is the number of input permutations) and a subset $I$ of $\{1,2,\dots,n\}$ as public parameters; we call the set $I$ as the \emph{fixing set} of our protocol.
By introducing the fixing set, we can, for example, use our secure grouping protocol starting from a permutation $(1,2)(3)(4) \cdots (n)$ and then shuffling all the numbers \emph{except the number $1$} (i.e., the random permutation $\sigma$ is chosen with constraint $\sigma(1) = 1$).
Such a generalized setting for the protocol here is required in our construction of the secure grouping protocol that flexibly covers various situations.

Let $\tau_1,\tau_2,\dots,\tau_k \in S_n$ be publicly known inputs for the protocol.
Then our permutation randomizing protocol with fixing set $I$ is performed as follows.
In the figures below, we consider an example where $n = 5$, $k = 2$, and $I = \{1,3\}$.
\begin{enumerate}
\item
Arrange $2k$ times the opened cards for numbers in $\{1,2,\dots,n\} \setminus I$ in increasing order, and face down the cards.
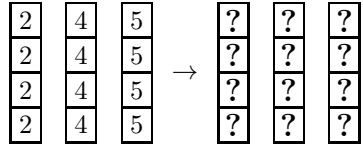
\begin{align*}
\begin{array}{ccc}
\card{2}  & \card{4} & \card{5} \\
\card{2}  & \card{4} & \card{5} \\
\card{2}  & \card{4} & \card{5} \\
\card{2}  & \card{4} & \card{5}
\end{array} 
\rightarrow
\begin{array}{ccc}
\card{\back}  & \card{\back} & \card{\back} \\
\card{\back}  & \card{\back} & \card{\back} \\
\card{\back}  & \card{\back} & \card{\back} \\
\card{\back}  & \card{\back} & \card{\back}
\end{array} 
\end{align*} 
\item
Apply Pile-Scramble Shuffle to the $2k$ rows simultaneously.
\item \label{item:generating_sigma}
For each of $2k$ rows, insert the opened cards for numbers in $I$ to the row in a way that the number card $\card{$a$}$ for $a \in I$ is at the $a$-th column.
Then face down all the inserted cards.
Note that the resulting committed card sequences represent the same (partially shuffled) permutation in $S_n$, say $\sigma$.
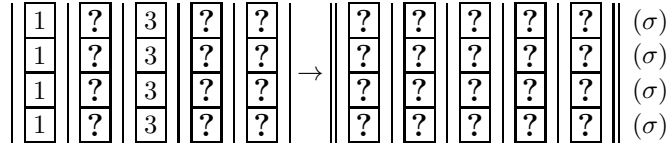
\begin{align*}
\begin{array}{|c|c|c|c|c|}
\card{1} & \card{\back} & \card{3} & \card{\back} & \card{\back} \\
\card{1} & \card{\back} & \card{3} & \card{\back} & \card{\back} \\
\card{1} & \card{\back} & \card{3} & \card{\back} & \card{\back} \\
\card{1} & \card{\back} & \card{3} & \card{\back} & \card{\back}
\end{array} 
\rightarrow
\begin{array}{||c|c|c|c|c||c}
\card{\back} & \card{\back} & \card{\back} & \card{\back} & \card{\back} & (\sigma) \\
\card{\back} & \card{\back} & \card{\back} & \card{\back} & \card{\back} & (\sigma) \\
\card{\back} & \card{\back} & \card{\back} & \card{\back} & \card{\back} & (\sigma) \\
\card{\back} & \card{\back} & \card{\back} & \card{\back} & \card{\back} & (\sigma)
\end{array} 
\end{align*} 
\item
For each $i = 1,2,\dots,k$, apply the permutation $\tau_i$ to one of the committed card sequences for $\sigma$ generated above.
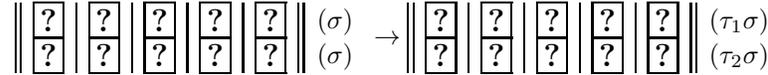
\begin{align*}
\begin{array}{||c|c|c|c|c||c}
\card{\back} & \card{\back} & \card{\back} & \card{\back} & \card{\back} & (\sigma) \\
\card{\back} & \card{\back} & \card{\back} & \card{\back} & \card{\back} & (\sigma)
\end{array} 
\rightarrow
\begin{array}{||c|c|c|c|c||c}
\card{\back} & \card{\back} & \card{\back} & \card{\back} & \card{\back} & (\tau_1 \sigma) \\
\card{\back} & \card{\back} & \card{\back} & \card{\back} & \card{\back} & (\tau_2 \sigma)
\end{array} 
\end{align*} 
\item
For each $i = 1,2,\dots,k$, perform the permutation division protocol for committed card sequences for $\sigma$ and $\tau_i \sigma$.
Then output the resulting sequences.
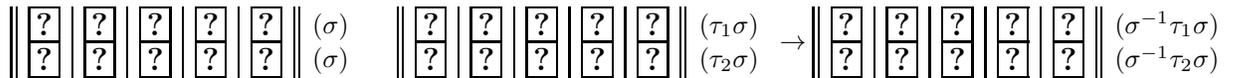
\begin{align*}
\begin{array}{||c|c|c|c|c||cc||c|c|c|c|c||c}
\card{\back} & \card{\back} & \card{\back} & \card{\back} & \card{\back} & (\sigma) & \ & \card{\back} & \card{\back} & \card{\back} & \card{\back} & \card{\back} & (\tau_1 \sigma) \\
\card{\back} & \card{\back} & \card{\back} & \card{\back} & \card{\back} & (\sigma) & \ & \card{\back} & \card{\back} & \card{\back} & \card{\back} & \card{\back} & (\tau_2 \sigma) \\
\end{array} 
\rightarrow
\begin{array}{||c|c|c|c|c||c}
\card{\back} & \card{\back} & \card{\back} & \card{\back} & \card{\back} & (\sigma^{-1} \tau_1 \sigma) \\
\card{\back} & \card{\back} & \card{\back} & \card{\back} & \card{\back} & (\sigma^{-1} \tau_2 \sigma)
\end{array} 
\end{align*} 
\end{enumerate}
We note that the (committed) permutation $\sigma$ generated in Step \ref{item:generating_sigma} is a uniformly random permutation in $S_n$ satisfying that $\sigma(j) = j$ for every $j \in I$.
For the security of the protocol, the following property is deduced straightforwardly from Proposition \ref{prop:security_permutation_division}.

\begin{prop}
\label{prop:security_permutation_randomizing}
The distribution of the only data available during the protocol, which is the $k$ card sequences opened during the permutation division protocols at Step 5, is uniform and is independent of the permutations $\sigma$ and $\sigma^{-1} \tau_i \sigma$ for $i = 1,2,\dots,k$.
\end{prop}

\section{Secure Grouping Protocol}\label{secure_grouping_protocol}
%本章では３章で紹介したプロトコルをどのように適用して秘匿グループ分けプロトコルを実現していくかを説明する。
In this section we present a construction of a secure grouping protocol, which is based on the permutation randomizing protocol described above.
See also \lq\lq Our Contributions'' paragraph in the introduction for an intuitive idea of our construction of the protocol.
\subsection{Our Setting for Grouping}
Before presenting our proposed secure grouping protocol, here we clarify our setting for the grouping problem.
We suppose that there are $n$ players, indexed by numbers $1,2,\dots,n$, to be divided into groups.
In our secure grouping protocol, the number of groups with $k$ members for each $k \geq 1$, denoted by $\mathcal{M}(k)$, should be determined in advance and is treated as public information.
Note that the integers $\mathcal{M}(k)$ satisfy that $\mathcal{M}(k) \geq 0$ for each $k \geq 1$ and $\sum_{k \geq 1} \mathcal{M}(k) = n$.
We may express $\mathcal{M}$ by the sequence $(\mathcal{M}(1),\mathcal{M}(2),\dots,\mathcal{M}(k))$ where $k$ is the maximal integer satisfying $\mathcal{M}(k) > 0$.

Our protocol can also handle a certain kind of constraints on the groupings, specified in the following manner.
For each integer $k \geq 1$, let $\mathcal{C}_k$ be a (possibly empty) set of non-empty subsets of $\{1,2,\dots,n\}$.
Let $\mathcal{C}$ be the sequence of $\mathcal{C}_1,\mathcal{C}_2,\dots$.
The meaning of a constraint $\mathcal{C}$ is the following:
\begin{itemize}
\item
For each $k \geq 1$ and each $C \in \mathcal{C}_k$, the players in $C$ must belong to the same group of size $k$.
\item
For any $k,k' \geq 1$, $C \in \mathcal{C}_k$, and $C' \in \mathcal{C}_{k'}$, if $C \neq C'$, then the players in $C$ and the players in $C'$ must belong to different groups.
\end{itemize}
Accordingly, the sets $\mathcal{C}_1,\mathcal{C}_2,\dots$ must satisfy the following conditions:
\begin{itemize}
\item
For any $C \in \mathcal{C}_k$, we have $1 \leq |C| \leq k$.
\item
For any $C \in \mathcal{C}_k$ and $C' \in \mathcal{C}_{k'}$ with $k \neq k'$, the subsets $C,C'$ of $\{1,2,\dots,n\}$ must be (different and) disjoint with each other.
\item
For any $C,C' \in \mathcal{C}_k$, $C$ and $C'$ must be disjoint unless these are equal.
\item
For any $k \geq 1$, we have $|\mathcal{C}_k| \leq \mathcal{M}(k)$.
\end{itemize}
Such a constraint $\mathcal{C}$ should also be specified in advance and is also treated as public information in our proposed protocol.

We define a \emph{grouping} of $n$ players to be a partition $\mathcal{G}$ of $\{1,2,\dots,n\}$, that is, a set of disjoint non-empty subsets of $\{1,2,\dots,n\}$ satisfying that the union of all sets in $\mathcal{G}$ is $\{1,2,\dots,n\}$.
For each $k \geq 1$, let $\mathcal{G}_k$ denote the (possibly empty) sets of all $A \in \mathcal{G}$ with $|A| = k$.
We say that a grouping $\mathcal{G}$ \emph{satisfies} a constraint $(\mathcal{M},\mathcal{C})$, if the followings hold:
\begin{itemize}
\item
We have $|\mathcal{G}_k| = \mathcal{M}(k)$ for any $k \geq 1$.
\item
If $k \geq 1$ and $C \in \mathcal{C}_k$, then there is a unique group $A$ in $\mathcal{G}_k$ satisfying $C \subset A$; we sometimes write this group $A$ as $A[C]$.
\item
If $k \geq 1$ and $C,C' \in \mathcal{C}_k$ are different, then we have $A[C] \neq A[C']$.
\end{itemize}
Note that the conditions for $\mathcal{C}$ and $\mathcal{M}$ introduced above ensure that the constraint can be satisfied by at least one grouping.
In our proposed secure grouping protocol, each player $P \in \{1,2,\dots,n\}$ will only receive the information on the (unique) set $A \in \mathcal{G}$ with $P \in A$; we sometimes write this group $A$ as $A[P]$.
We give some examples of the situation above for the sake of explanation.

\begin{exam}
We consider a case of grouping of nine players into three groups with three members, with constraints that Players $8$ and $9$ want to be in the same group while Player $1$ does not want to be in the same group as them.
This situation can be expressed by $\mathcal{M} = (0,0,3)$, $\mathcal{C}_1 = \mathcal{C}_2 = \emptyset$, and $\mathcal{C}_3 = \{\{1\},\{8,9\}\}$.
Then an example of a grouping is given by $\mathcal{G} = \mathcal{G}_3 = \{\{1,4,6\},\{2,5,7\},\{3,8,9\}\}$.
\end{exam}

\begin{exam}
\label{ex:Werewolf_situation_simple}
We consider a situation to classify five players into two distinguished persons and three ordinary persons in the following manner: each distinguished person is told who is the other distinguished person; while each ordinary person is not told who are the distinguished persons, nor who are the other ordinary persons.
This situation can be realized by treating each of the three ordinary persons as an individual group of size one consisting of him/herself alone, while treating the two distinguished persons naturally as a (unique) group of size two.
Accordingly, we set $\mathcal{M} = (3,1)$ and set each $\mathcal{C}_k$ to be an empty set.
Then an example of a grouping is given by $\mathcal{G} = \{\{2\},\{4\},\{5\},\{1,3\}\}$ (hence $\mathcal{G}_1 = \{\{2\},\{4\},\{5\}\}$ and $\mathcal{G}_2 = \{\{1,3\}\}$); this means that Players $2$, $4$, and $5$ are ordinary persons, and Players $1$ and $3$ are the distinguished persons.
%
%We consider a simple situation for Werewolf game, where five players are classified into two \lq\lq Werewolves'' and three \lq\lq Villagers''.
%The additional requirements are: each Werewolf is told who is the other Werewolf; while each Villager is not told who are Werewolves, nor who are the other Villagers.
%This situation can be realized by treating each of the three Villagers as an individual group of size one consisting of him/herself alone, while treating the two Werewolves naturally as a (unique) group of size two.
%Accordingly, we set $\mathcal{M} = (3,1)$ and set each $\mathcal{C}_k$ to be an empty set.
%Then an example of a grouping is given by $\mathcal{G} = \{\{2\},\{4\},\{5\},\{1,3\}\}$ (hence $\mathcal{G}_1 = \{\{2\},\{4\},\{5\}\}$ and $\mathcal{G}_2 = \{\{1,3\}\}$); this means that Players $2$, $4$, and $5$ are Villagers, and Players $1$ and $3$ are Werewolves.
\end{exam}

\begin{exam}
\label{ex:Werewolf_situation_extended}
We consider a slightly more complicated situation where seven players are classified into two \lq\lq Role A'' players, one \lq\lq Role B'' player, two \lq\lq Role C'' players, and two ordinary players.
The additional requirements are as follows:
\begin{itemize}
\item
Each player with Role A and each player with Role C are told his/her own role, are told who is the other player with the same role as him/herself, but are told nothing about the remaining players' roles.
\item
The player with Role B and each ordinary player are told his/her own role, but are told nothing about the remaining players' roles.
\end{itemize}
In contrast to Example \ref{ex:Werewolf_situation_simple} where the ordinary and the distinguished persons can be distinguished just by the sizes of the groups (one for the former, and two for the latter), here we should distinguish Role B from the ordinary players (both would be represented by size-one groups) and Role C from Role A (both would be represented by size-two groups).

A solution is to introduce dummy indices $8$ representing \lq\lq Role B'' and $9$ representing \lq\lq Role C''.
Namely, we divide the nine numbers into one group consisting of the dummy index $8$ and a player's index (who becomes \lq\lq Role B''), one group consisting of the dummy index $9$ and two players' indices (who become \lq\lq Role C''), two groups consisting of a player's index only (who becomes \lq\lq ordinary player''), and one group consisting of two players' indices only (who become \lq\lq Role A'').
Accordingly, we set the constraint to be $\mathcal{M} = (2,2,1)$, $\mathcal{C}_1 = \emptyset$, $\mathcal{C}_2 = \{\{8\}\}$, and $\mathcal{C}_3 = \{\{9\}\}$.
An example of a grouping is given by $\mathcal{G}_1 = \{\{1\},\{6\}\}$, $\mathcal{G}_2 = \{\{2,7\},\{4,8\}\}$, and $\mathcal{G}_3 = \{\{3,5,9\}\}$; this means that Players $1$ and $6$ are ordinary players, Players $2$ and $7$ are the Role A players, Player $4$ is the Role B player, and Players $3$ and $5$ are the Role C players.
We note that similar ideas to introduce dummy indices representing \lq\lq names of groups'' can be applied to the case of more complicated groupings.
\end{exam}
\subsection{Secure Grouping Protocol for Simpler Case}\label{subsection:grouping_protocol_two}

Before describing our proposed secure grouping protocol in a general form, here we consider a simpler case with empty constraints, that is, the sets $\mathcal{C}_k$ for specifying constraints for the groupings are all empty.
This case includes the case mentioned in Example \ref{ex:Werewolf_situation_simple} above.

Here we suppose that the number $n$ of players for the secure grouping and the group size function $\mathcal{M}$ (as well as the empty constraint sets $\mathcal{C}_k$) are determined in advance and are public information.
As a pre-computation part of the protocol, the players compute a permutation $\tau \in S_n$ as follows; note that this $\tau$ is also a public information, therefore the computation of $\tau$ does not need any secure computation protocol.
Let $k$ denote the maximal integer with $\mathcal{M}(k) > 0$.
First, the players compute integers $a_0,a_1,\dots,a_{k-1}$ recursively by $a_0 := 0$ and $a_i := a_{i-1} + i \cdot \mathcal{M}(i)$ for $1 \leq i \leq k-1$.
Then the players define $\tau$ to be the product of cyclic permutations
\[
(a_{i-1} + (j-1)i + 1\quad a_{i-1} + (j-1)i + 2\quad \cdots\quad a_{i-1} + (j-1)i + i)
\]
for all $1 \leq i \leq k$ and $1 \leq j \leq \mathcal{M}(i)$.
We note that this permutation $\tau$ is of type $\langle r_1{}^{\mathcal{M}(r_1)},r_2{}^{\mathcal{M}(r_2)},\dots,r_{\ell}{}^{\mathcal{M}(r_{\ell})} \rangle$ where $r_1,r_2,\dots,r_{\ell}$ are the integers at which the function $\mathcal{M}$ takes a positive value.
For example, if $\mathcal{M} = (3,2,0,1)$, then we have
\[
\tau
= (1)(2)(3)(4\ 5)(6\ 7)(8\ 9\ 10\ 11)
= (4\ 5)(6\ 7)(8\ 9\ 10\ 11)
\in \langle 1^3,2^2,4^1 \rangle \enspace.
\]

We also note that, our protocol below utilizes the permutation randomizing protocol introduced in Section \ref{permutation_randomizing_protocol} with empty fixing set $I = \emptyset$ as a sub-protocol.
This sub-protocol is given a number of publicly known permutations $\tau_1,\tau_2,\dots,\tau_{\ell} \in S_n$ as inputs, and outputs committed card sequences for permutations $\rho_1,\rho_2,\dots,\rho_{\ell} \in S_n$, where $\rho_i = \sigma^{-1} \tau_i \sigma$ with common and uniformly random permutation $\sigma \in S_n$ for each $1 \leq i \leq \ell$.

Then, given the data above including the permutation $\tau$, the main part of our secure grouping protocol is executed as follows, where $k$ denotes as above the maximal integer with $\mathcal{M}(k) > 0$ (which is equal to the maximal length of cyclic permutations involved in $\tau$):
\begin{enumerate}
\item
The players jointly execute the permutation randomizing protocol (with empty fixing set $I = \emptyset$) for input permutations $\tau,\tau^2,\dots,\tau^{k-1}$, and obtain the committed card sequences $x[\rho],x[\rho^2],\dots,x[\rho^{k-1}]$ for permutations $\rho,\rho^2,\dots,\rho^{k-1}$ with $\rho = \sigma^{-1} \tau \sigma$ (note that $\sigma^{-1} \tau^j \sigma = (\sigma^{-1} \tau \sigma)^j$ for any $j$).
\item
Each Player $i$ picks the $i$-th card $x[\rho^j]_i$ of the card sequence $x[\rho^j]$ for all $1 \leq j \leq k-1$.
Then the numbers (except the number $i$ itself) written on the front of these $k-1$ cards (that may be duplicated) show the other players in Player $i$'s group.
\end{enumerate}

For example, if $\tau \in S_{11}$ is as above and $\sigma = (1\ 8)(2\ 6\ 3\ 7\ 10)(4\ 11) \in S_{11}$ is chosen in the protocol, then we have $k = 4$, $\rho = (1\ 9\ 7\ 4)(2\ 3)(5\ 11)$, and the card sequences satisfy
\[
\begin{split}
\mbox{fronts of } x[\rho] &= \card{4}\ \card{3}\ \card{2}\ \card{7}\ \card{11}\ \card{6}\ \card{9}\ \card{8}\ \card{1}\ \card{10}\ \card{5} \enspace,\\
\mbox{fronts of } x[\rho^2] &= \card{7}\ \card{2}\ \card{3}\ \card{9}\ \card{5}\ \card{6}\ \card{1}\ \card{8}\ \card{4}\ \card{10}\ \card{11} \enspace,\\
\mbox{fronts of } x[\rho^3] &= \card{9}\ \card{3}\ \card{2}\ \card{1}\ \card{11}\ \card{6}\ \card{4}\ \card{8}\ \card{7}\ \card{10}\ \card{5} \enspace.
\end{split}
\]
Then Player $3$ takes the cards $\card{2}$\,, $\card{3}$\,, and $\card{2}$\,, therefore the player's group is $\{2,3\}$.
On the other hand, Player $4$ takes the cards $\card{7}$\,, $\card{9}$\,, and $\card{1}$\,, therefore the player's group is $\{1,4,7,9\}$.

\subsection{Secure Grouping Protocol for General Case}\label{subsection:grouping_protocol_general}

From now, we describe our secure grouping protocol in a general case where the constraint set $\mathcal{C}_k$ may be non-empty.
We note that these sets $\mathcal{C}_k$ are also determined in advance and publicly known.
Now the pre-computation part to determine a public permutation $\tau \in S_n$ is executed as follows, where $k$ denotes the maximal integer with $\mathcal{M}(k) > 0$:
\begin{itemize}
\item
Initialize $\tau$ and auxiliary counters $B$ by $\tau \leftarrow \mathsf{id}_n$ and $B \leftarrow \{1,2,\dots,n\} \setminus \bigcup_{j=1}^{k} \bigcup_{A \in \mathcal{C}_j} A$.
Then do the following for each $\lambda = 1,2,\dots,k$:
\begin{itemize}
\item
Do the following for each $\mu = 1,2,\dots,\mathcal{M}(\lambda)$:
\begin{itemize}
\item
If $\mathcal{C}_{\lambda}$ contains a set, say $C = \{a_1,a_2,\dots,a_{\ell}\}$, then update $\tau$ and $B$ by $\tau \leftarrow \tau \cdot (a_1\ a_2\ \cdots\ a_{\ell}\ b_1\ b_2\ \cdots\ b_{\lambda - \ell})$ and $B \leftarrow B \setminus \{b_1,b_2,\dots,b_{\lambda - \ell}\}$, where $b_1,b_2,\dots,b_{\lambda - \ell}$ are the first $\lambda - \ell$ elements of the set $B$; and then remove the set $C$ from $\mathcal{C}_{\lambda}$.
\item
If $\mathcal{C}_{\lambda}$ is empty, then update $\tau$ and $B$ by $\tau \leftarrow \tau \cdot (b_1\ b_2\ \cdots\ b_{\lambda})$ and $B \leftarrow B \setminus \{b_1,b_2,\dots,b_{\lambda}\}$, where $b_1,b_2,\dots,b_{\lambda}$ are the first $\lambda$ elements of the set $B$.
\end{itemize}
\end{itemize}
\end{itemize}
This procedure is constructed to ensure that the resulting $\tau$ is a permutation in $S_n$ and satisfies the constraint $(\mathcal{M},\mathcal{C})$.
For example, if $n = 9$, $\mathcal{M} = (2,2,1)$ and $\mathcal{C}$ are as in Example \ref{ex:Werewolf_situation_extended}, then the computation above is performed as follows:
\begin{quote}
(Initialize) $\tau = \mathsf{id}_9$, $\mathcal{C}_1 = \emptyset$, $\mathcal{C}_2 = \{\{8\}\}$, $\mathcal{C}_3 = \{\{9\}\}$, $B = \{1,2,3,4,5,6,7\}$ \\
$\to$ ($\lambda = 1$, $\mu = 1$) $\tau = (1) = \mathsf{id}_9$, $\mathcal{C}_1 = \emptyset$, $\mathcal{C}_2 = \{\{8\}\}$, $\mathcal{C}_3 = \{\{9\}\}$, $B = \{2,3,4,5,6,7\}$ \\
$\to$ ($\lambda = 1$, $\mu = 2$) $\tau = (2) = \mathsf{id}_9$, $\mathcal{C}_1 = \emptyset$, $\mathcal{C}_2 = \{\{8\}\}$, $\mathcal{C}_3 = \{\{9\}\}$, $B = \{3,4,5,6,7\}$ \\
$\to$ ($\lambda = 2$, $\mu = 1$) $\tau = (8\ 3)$, $\mathcal{C}_1 = \mathcal{C}_2 = \emptyset$, $\mathcal{C}_3 = \{\{9\}\}$, $B = \{4,5,6,7\}$ \\
$\to$ ($\lambda = 2$, $\mu = 2$) $\tau = (8\ 3)(4\ 5)$, $\mathcal{C}_1 = \mathcal{C}_2 = \emptyset$, $\mathcal{C}_3 = \{\{9\}\}$, $B = \{6,7\}$ \\
$\to$ ($\lambda = 3$, $\mu = 1$) $\tau = (8\ 3)(4\ 5)(9\ 6\ 7)$, $\mathcal{C}_1 = \mathcal{C}_2 = \mathcal{C}_3 = \emptyset$, $B = \emptyset$
\end{quote}

We also note that, our protocol below utilizes the permutation randomizing protocol with fixing set $I = \bigcup_{j=1}^{k} \bigcup_{A \in \mathcal{C}_j} A \subset \{1,2,\dots,n\}$ introduced in Section \ref{permutation_randomizing_protocol}.
This sub-protocol is given publicly known permutations $\tau_1,\tau_2,\dots,\tau_{\ell} \in S_n$ as inputs, and outputs committed card sequences for permutations $\rho_1,\rho_2,\dots,\rho_{\ell} \in S_n$, where for each $i$, $\rho_i = \sigma^{-1} \tau_i \sigma$ with common and uniformly random permutation $\sigma \in S_n$ satisfying that $\sigma(a) = a$ for every $a \in I$.

Then, given the data above including the permutation $\tau$, the main part of our secure grouping protocol is executed as follows, where $k$ denotes as above the maximal integer with $\mathcal{M}(k) > 0$:
\begin{enumerate}
\item
The players jointly execute the permutation randomizing protocol with fixing set $I$ for input permutations $\tau,\tau^2,\dots,\tau^{k-1}$, and obtain the committed card sequences $x[\rho],x[\rho^2],\dots,x[\rho^{k-1}]$ for permutations $\rho,\rho^2,\dots,\rho^{k-1}$ with $\rho = \sigma^{-1} \tau \sigma$ (note that $\sigma^{-1} \tau^j \sigma = (\sigma^{-1} \tau \sigma)^j$ for any $j$).
\item
Each Player $i$ picks the $i$-th card $x[\rho^j]_i$ of the card sequence $x[\rho^j]$ for all $1 \leq j \leq k-1$.
Then the numbers (except the number $i$ itself) written on the front of these $k-1$ cards (that may be duplicated) show the other players in Player $i$'s group.
\end{enumerate}
We note that, if $\mathcal{C}_i = \emptyset$ for any $1 \leq i \leq k$, then the protocol above coincides with the protocol described in Section \ref{subsection:grouping_protocol_two}.

\section{Proofs of Correctness and Security}
In this section, we prove the correctness and the security of our proposed secure grouping protocol.

\subsection{Proof of Correctness}
In this subsection, we prove the correctness of our secure grouping protocol as follows:
\begin{theo}
Let $(\mathcal{M},\mathcal{C})$ be a possible constraint for our protocol.
Then our secure grouping protocol with constraint $(\mathcal{M},\mathcal{C})$ generates each grouping $\mathcal{G}$ satisfying the constraint $(\mathcal{M},\mathcal{C})$ with equal probability.
\end{theo}

To prove the theorem, we introduce some auxiliary definitions.
First, let $\pi \in S_n$ be a permutation and let $\pi = \pi_1 \pi_2 \cdots \pi_{\ell}$ be the decomposition of $\pi$ into disjoint cyclic permutations $\pi_1,\dots,\pi_{\ell}$, where the cyclic permutations of length $1$ are also included in the decomposition.
Then we define the grouping $\mathcal{G}[\pi]$ specified by $\pi$ to be the set of the cyclic areas of the cyclic permutations $\pi_1,\pi_2,\dots,\pi_{\ell}$.
For example, if $\pi = (1\,5)(4)(2\,6\,3) \in S_6$, then $\mathcal{G}[\pi] = \{\{4\},\{1,5\},\{2,3,6\}\}$.

Secondly, we say that a permutation $\pi \in S_n$ satisfies the constraint $(\mathcal{M},\mathcal{C})$, if the following conditions are satisfied:
\begin{itemize}
\item
Let $r_1 < r_2 < \cdots < r_L$ be all the positive integers with $\mathcal{M}(r_i) > 0$.
Then $\pi \in \langle r_1{}^{\mathcal{M}(r_1)},r_2{}^{\mathcal{M}(r_2)},\dots,r_L{}^{\mathcal{M}(r_L)} \rangle$.
\item
Let $k \geq 1$ and $C = \{a_1,a_2,\dots,a_h\} \in \mathcal{C}_k$ (we assume that the elements $a_1,a_2,\dots,a_h$ of any set $C \in \mathcal{C}_k$ are always written in increasing order, in our argument below as well as the construction of the secure grouping algorithm).
Then the numbers $a_1,a_2,\dots,a_h$ are involved in the cyclic area of the same cyclic permutation in the decomposition of $\pi$, say $\pi_i$, and we have $\pi_i(a_j) = a_{j+1}$ for any $1 \leq j \leq h-1$.
\end{itemize}
We note that, if $\pi \in S_n$ satisfies the constraint $(\mathcal{M},\mathcal{C})$, then the grouping $\mathcal{G}[\pi]$ satisfies the constraint $(\mathcal{M},\mathcal{C})$ as well.
We note also that, by the construction, the permutation $\tau \in S_n$ computed in the pre-computation part of our secure grouping protocol with constraint $(\mathcal{M},\mathcal{C})$ satisfies the constraint $(\mathcal{M},\mathcal{C})$ in the sense above.

Now we show the following property:

\begin{lem}
\label{lem:correspondence_of_permutation_and_grouping}
Let $\rho \in S_n$ be the permutation generated (in the committed form) in our secure grouping protocol.
Then the output of our secure grouping protocol is $\mathcal{G}[\rho]$.
\end{lem}
\begin{proof}
Let $k$ denote the integer specified in the construction of the protocol.
Let $i \in \{1,2,\dots,n\}$, and let $\rho_i$ denote the cyclic permutation in the decomposition of $\rho$ whose cyclic area contains $i$.
Then, by the definition of the card sequence representing a permutation, the numbers written on the cards obtained by Player $i$ at the end of the protocol are $(\rho^j)^{-1}(i) = \rho^{-j}(i) = \rho_i{}^{-j}(i)$ for $j = 1,2,\dots,k-1$.
Moreover, by the definition of $k$, the length of the cyclic permutation $\rho_i$ is at most $k$, therefore the set of those numbers $\rho_i{}^{-j}(i)$ for $j = 1,2,\dots,k-1$ together with the number $i$ itself is equal to the group in $\mathcal{G}[\rho]$ containing $i$, the latter being the cyclic area of $\rho_i$ by definition.
Hence the claim holds.
\end{proof}

By Lemmas \ref{lem:expression_of_conjugation}, \ref{lem:correspondence_of_permutation_and_grouping} and the fact that the (partially shuffled) permutation $\sigma \in S_n$ generated in the permutation randomizing protocol fixes each element of the fixing set $I$, it follows that the output of our secure grouping algorithm is a grouping satisfying the given constraint $(\mathcal{M},\mathcal{C})$.

Moreover, since the permutation $\sigma \in S_n$ generated in the permutation randomizing protocol is chosen uniformly at random from all the permutations in $S_n$ that fixes every element of $I$, the following property is deduced straightforwardly by Lemma \ref{lem:expression_of_conjugation}:

\begin{lem}
Given the input $\tau,\tau^2,\dots,\tau^{k-1}$ for the permutation randomizing protocol executed internally in our secure grouping protocol, the (committed) permutations $\rho,\rho^2,\dots,\rho^{k-1}$ corresponding to the output of the permutation randomizing protocol satisfy that $\rho$ is uniformly random over the set of all permutations in $S_n$ satisfying the constraint $(\mathcal{M},\mathcal{C})$.
\end{lem}

On the other hand, the following property is deduced straightforwardly by the definition of the grouping $\mathcal{G}[\pi]$ specified by a permutation $\pi$:

\begin{lem}
Let $(\mathcal{M},\mathcal{C})$ be a given constraint.
For any grouping $\mathcal{G}$ satisfying the constraint $(\mathcal{M},\mathcal{C})$, the number of permutations $\pi$ that satisfies the constraint $(\mathcal{M},\mathcal{C})$ and satisfies $\mathcal{G}[\pi] = \mathcal{G}$ is independent of the choice of $\mathcal{G}$.
\end{lem}

Now our claim follows by combining the last two lemmas: Namely, for any two groupings $\mathcal{G},\mathcal{G}'$ satisfying the constraint $(\mathcal{M},\mathcal{C})$, the number of permutations $\rho$ satisfying the constraint $(\mathcal{M},\mathcal{C})$ that specifies the grouping $\mathcal{G}$ is equal to the number of those permutations that specifies the grouping $\mathcal{G}'$, and those permutations $\rho$ are chosen with equal probability.
This completes the proof.

\subsection{Proof of Security}
In this subsection, we prove the security of our secure grouping protocol as follows:
\begin{theo}
Let $(\mathcal{M},\mathcal{C})$ be a possible constraint for our secure grouping protocol.
Let $\mathcal{G}$ denote the grouping which is the output of our secure grouping protocol with constraint $(\mathcal{M},\mathcal{C})$.
Then, for any Player $i$, the information obtained by the player during the protocol is independent of the groups $A \in \mathcal{G}$ that do not contain $i$.
\end{theo}

To prove the theorem, we first note that, the argument in the proof of Lemma \ref{lem:correspondence_of_permutation_and_grouping} implies that the output of Player $i$ in the secure grouping protocol is the sequence of numbers $(\rho^{-1}(i),\rho^{-2}(i),\dots,\rho^{-(k-1)}(i))$, where $\rho$ is the permutation generated (in the committed form) in the protocol.
Let $\rho_i$ denote the unique cyclic permutation involved in $\rho$ that contains the number $i$.
Then, by using the output above, Player $i$ can recover not only the cyclic area of $\rho_i$ (which is an \emph{unordered} set) but also the whole of the cyclic permutation $\rho_i$ itself.
Therefore, the information obtained by Player $i$ during the protocol is the cyclic permutation $\rho_i$ as well as the card sequences that are opened during the permutation randomizing protocol.
Moreover, Proposition \ref{prop:security_permutation_randomizing} implies that the latter cards opened during the permutation randomizing protocol provides essentially no information, therefore it suffices to concern the information on the cyclic permutation $\rho_i$ only.

Now the following property is deduced straightforwardly by the definition of the grouping $\mathcal{G}[\pi]$ specified by a permutation $\pi$:

\begin{lem}
\label{lem:other_component_of_grouping}
Let $i$ and $\rho_i$ be as above.
Let $\mathcal{G}'$ and $\mathcal{G}''$ be any grouping satisfying the constraint $(\mathcal{M},\mathcal{C})$, in which the group including $i$ is equal to the cyclic area of $\rho_i$.
Then, among the permutations $\widetilde{\rho}$ whose decomposition into disjoint cyclic permutations involves $\rho_i$, the number of those permutations that satisfies $\mathcal{G}[\widetilde{\rho}] = \mathcal{G}'$ is equal to the number of those permutations $\widetilde{\rho}$ that satisfies $\mathcal{G}[\widetilde{\rho}] = \mathcal{G}''$.
\end{lem}

Since the choice of the permutation $\rho$ is uniformly random, it follows by Lemma \ref{lem:other_component_of_grouping} that the conditional distribution of the grouping $\mathcal{G}$ generated by our secure grouping algorithm, except the group including $i$, conditioned on the choice of the cyclic permutation $\rho_i$ is still the uniform distribution.
This completes the proof.

\paragraph{\bf Acknowledgement}
We thank the members of Shin-Akarui-Angou-Benkyou-Kai for their helpful comments.

\bibliographystyle{abbrv}
\bibliography{cardbib}

\end{document}